\documentclass[11pt]{article}
\usepackage{graphicx,hyperref}
\usepackage{amssymb,amsmath,amsthm}
\usepackage{fullpage}
\usepackage[left]{lineno}

\newcommand{\comment}[1]{} %{\large\bf Comments \small\bf #1}}

\newtheorem{theorem}{Theorem}
\newtheorem{definition}[theorem]{Definition}
\newtheorem{remark}[theorem]{Remark}
\newtheorem{lemma}[theorem]{Lemma}

\def\R{{\mathbb R}}
\begin{document}
\date{}     

\title{Computing Melodic Templates in Oral Music Traditions}
% \thanks{This research has received funding from
%the projects COFLA2 (Junta de Andaluc\'ia, P12-TIC-1362), GALGO (Spanish Ministry of Economy and Competitiveness, MTM2016-76272-R AEI/FEDER,UE) and CONNECT (EU-H2020/MSCA under grant agreement 2016-734922).}}
% \thanks{This research has received funding from
% the projects COFLA2 (Junta de Andaluc\'ia, P12-TIC-1362), GALGO (Spanish Ministry of Economy and Competitiveness, MTM2016-76272-R AEI/FEDER,UE) and CONNECT (EU-H2020/MSCA under grant agreement 2016-734922).}

\author{Sergey Bereg\thanks{Department of Computer Science, University of Texas at Dallas, USA. Email: besp@utdallas.edu.}
\and
Jos\'e-Miguel D\'iaz-B\'a\~nez\thanks{Department of Applied Mathematics II, University of Seville, Spain. Email: \{dbanez,nkroher,iventura\}@us.es. 
%\begin{figure}[h!]
%\begin{minipage}[l]{0.3\textwidth}
%\includegraphics[trim=10cm 6cm 10cm 5cm,clip,scale=0.15]{eu_logo.pdf}
%\end{minipage}\hspace{-2cm}
%\begin{minipage}[l][1cm]{0.9\textwidth}
This work has received funding from Spanish Ministry of Economy and Competitiveness (MTM2016-76272-R AEI/FEDER,UE), the Junta de Andaluc\'ia (Cofla project, P12-TIC-1362) and the European Union's Horizon 2020 research and innovation programme under the Marie Sk\l{}odowska-Curie grant agreement No 734922.
%\end{minipage}
%\end{figure}
%\vspace{-0.6cm}
} 
\and
Nadine Kroher$^\dag$
\and
Inmaculada Ventura$^\dag$
}

% \title{Computing Melodic Templates in Oral Traditions: A geometrical approach}

% \author{Sergey Bereg\thanks{
% Department of Computer Scence, University of Texas at Dallas, USA}
% \and
% {Jos\'e-Miguel D\'iaz-B\'a\~nez}\thanks{
% Department of Applied Mathematics II, University of Seville, Spain}
% \and
% {Nadine Kroher}\footnotemark[3]
% \and
% {Inmaculada Ventura}\footnotemark[3]
% }

\maketitle

\begin{abstract}
The term \textit{melodic template or skeleton} refers to a basic melody which is subject to variation during a music performance. In many oral music tradition, these templates are implicitly passed throughout generations without ever being formalized in a score. In this work, we introduce a new geometric optimization problem, the spanning tube problem, to approximate a melodic template for a set of labeled performance transcriptions corresponding to an specific style in oral music traditions. 
Given a set of $n$ piecewise linear functions, we solve the problem of finding a continuous function, $f^*$, and a minimum value, $\varepsilon^*$, such that, the vertical segment of length $2\varepsilon^*$ centered at $(x,f^*(x))$ intersects at least $p$ functions ($p\leq n$).  
 The method explored here also provide a novel tool for quantitatively assess the amount of melodic variation which occurs across performances.
\end{abstract}

\section{Introduction}
There exists a long tradition of using mathematic concepts to understand, model and create music. Ever since ancient Greek mathematicians have first formulated musical acoustics from a mathematical standpoint \cite{crocker1963pythagorean}, music and mathematics have gone hand in hand (see i.e. \cite{vaughn2000music} for an overview). As a result, a number of challenging multi-disciplinary fields have opened up, including algorithmic composition \cite{nierhaus2009algorithmic}, ethnomathematics \cite{chemillier2002ethnomusicology} and automatic music analysis \cite{madden1999fractals}. 

A challenging issue in the mathematical research of music is to
understand the nature of music. A deep discussion about the use of mathematics and computer science in music can be found in \cite{volk2012}. Since computers allow to efficiently process musical data, mathematical research in music has also provided theoretical foundations for music technology, an interdisciplinary science of computational description, analysis and processing of music and music data. An important subfield of music technology is music information retrieval (MIR), the scientific discipline of retrieving high-level information from music recordings.
In particular, several optimization problems have been explored in the context of music analysis, including the modeling of rhythmic structures \cite{godfried-2010, aichholzer2015, barba2016}, harmony improvisation \cite{valian2014}, melodic similarity \cite{typke-03} and the detection of melodic patterns \cite{lubiw2004pattern}. 

This paper addresses a geometric problem inspired by musical characteristics found in \textit{flamenco} music. \textit{Flamenco} is a rich oral music tradition from southern Spain, which attracts a growing community of enthusiasts around the globe. 
The computational analysis of \textit{flamenco} melodies, which has mainly been carried out in the MIR community, has revealed algorithmic challenges due to its high degree of ornamentation and improvisation \cite{flamencoStyle2, pat, patternRepetition}. We propose a geometric method to approximate the melodic template inherent to a set of \textit{flamenco} performances. The term \textit{melodic template} refers to a basic melodic line, which undergoes variation, i.e. in the form of ornamentation and embellishment, during a music performance. An example of MIDI transcriptions of the different interpretations of the same underlying template is depicted in Figure \ref{fig:example}. Similar concepts are found in \textit{Arab-andalusian} \cite{kroher2017}, \textit{Carnatic} \cite{rao2014} and \textit{Shizou} \cite{sizhu1985} music. 

Mathematically speaking, we are interested in approximating a set of $n$ piecewise linear functions $f_i,i=1, \cdots, n$ representing melodies in a  temporal-pitch domain by a polygonal curve $f^*$ so that, in any time, $f^*$ is similar to an acceptable amount of such melodies. 

Related problems are polygonal approximation and fitting problems. The aim of polygonal approximation is to approximate a complex polygonal curve by a simpler polygonal curve so that total approximation error is minimized for a given error metric. This is an important task in computer vision, computer graphics, digital cartography, statistics, and data compression. Polygonal approximation is an old problem and it has been well-studied from various perspectives as
geographic information systems \cite{douglas1973}; digital
image analysis \cite{imai1986,hobby1993}; and computational geometry \cite{chan1996}. On the other hand, the fitting problems addresses the approximation of a discrete set of points by a polygonal curve or step function with a bounded number of line segments minimizing a fitting measure \cite{diaz2001,fournier2011}. 

Unlike polygonal approximation and fitting problems, the input of our problem is a set of polygonal curves instead of a single polygonal curve or a set of points, and the constraint of the output is given by the number of neighboring curves at any time in the time-pitch domain. To the best of our knowledge, this problem has not been considered in the literature and it is motivated by some potential applications in the automatic indexing and analysis of oral music traditions. In the sequel, we outline several such applications.

\begin{figure}[ht]
\centering
\includegraphics[width=1\columnwidth]{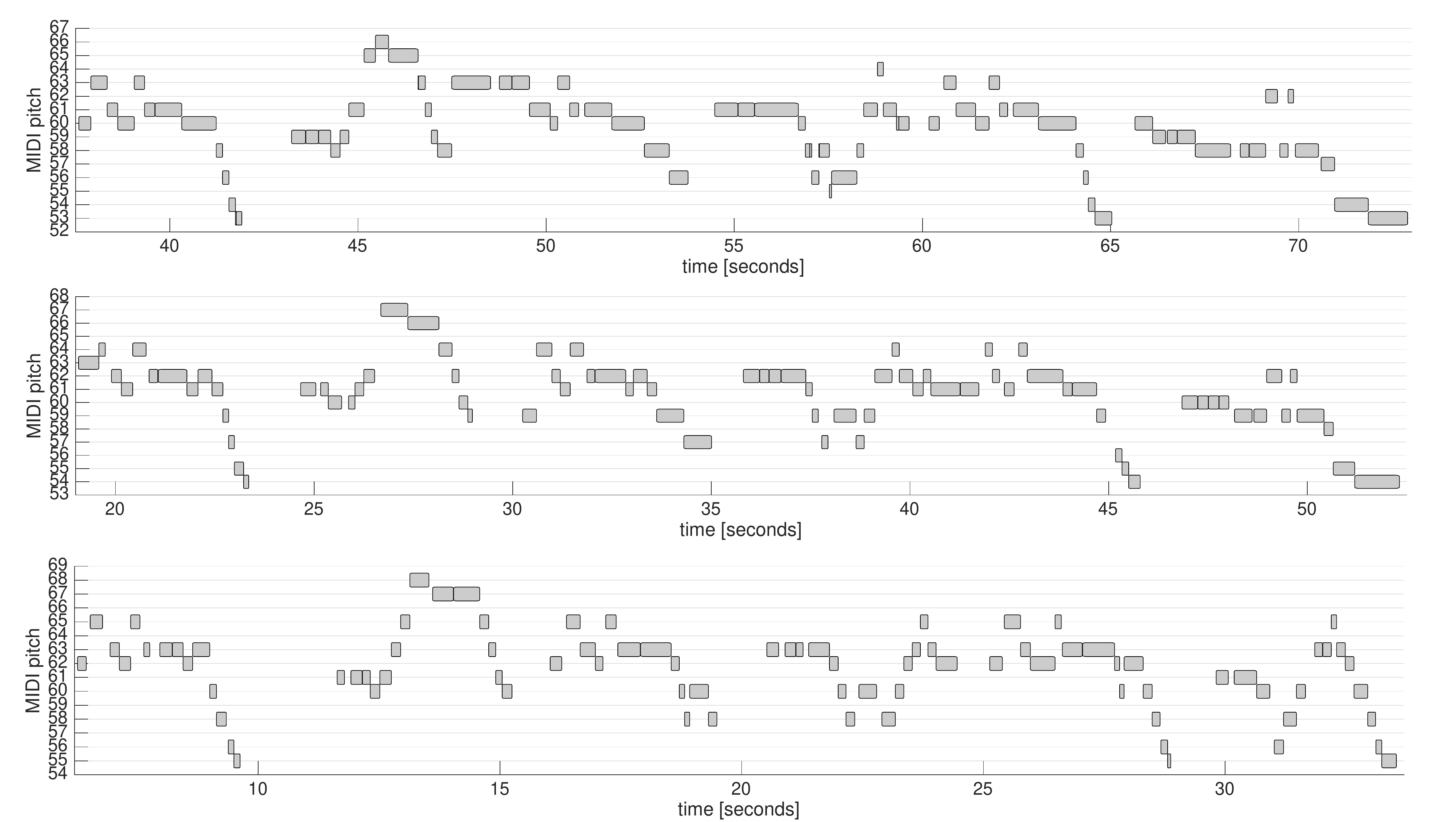}
\caption{MIDI transcriptions of three performances of the \textit{fandango de Valverde} melody. 
%The thick gray line represents a manually annotated skeleton.
}
\label{fig:example}
\end{figure}

Comparing a performance to its underlying template, allows researchers to explore the creative process itself, model expressiveness and study the evolution of improvisation and interpretation \cite{widmer2004computational}. While these aspects have been extensively studied in popular, jazz and classical music, existing approaches rely on the existence of the template as a musical score. However, in oral music traditions, scores are usually not available and rare manual transcriptions refer to a particular performance instead of the underlying melodic line. Being able to automatically approximate a template based on a set of performance transcriptions is a fundamental requirement for large-scale comparative performance analysis for oral music traditions

Another potential application for automatically generated melodic templates is melody classification, where the task is to assign a label to an unknown melody according to its underlying melodic template. In the context of flamenco music, due to the unavailability of the template, this problem has been previously addressed in a supervised $k$-nearest neighbor classification scenario: An unknown melody is tagged based on the labels of the $k$ most similar melodies in a manually annotated dataset. However, the necessary pair-wise comparisons with a number of melodies at runtime represent a major limitation for scalability. Comparing to automatically extracted templates instead, can decrease the runtime complexity from $k$ to a single comparison.

The structure of the paper is as follows. 
We discuss the melody modeling and data preprocessing in Section \ref{sec:preprocessing}. 
In Section \ref{sec:STP} we introduce a new problem called the {\em Spanning Tube Problem} and develop a method for solving it.
In Section \ref{sect:case} we apply our method for a comparative performance analysis of the four \textit{fandango} styles investigated in this study.
Finally, we summarize our research and future directions in Section \ref{sect:conclusion}.

\section{Melody modeling: dataset and preprocessing}\label{sec:preprocessing}

In the scope of this study, we analyze a corpus of 40 flamenco commercial recordings belonging to 4 variants of the \textit{fandango de Huelva} style. All interpretations of a variant are characterized by a common melodic skeleton which is subject to ornamentation and variation during performance. An automatic transcription tool \cite{CANTE} was used to extract a note representation of the singing voice melody and transcription errors were subsequently corrected by flamenco experts. 

As a first preprocessing step, we convert each transcription, which is a discrete set of notes, described by their onset time, duration and MIDI pitch value, to a point set representation $A = \{a_1, a_2, .., a_m\}$ where a note $a_i = \{(x_i,y_i)\}$ is characterized by its onset time relative to the last note onset, $x_i \in [0,1]$, and its MIDI pitch value, $y_i$. Given that the transcription tool does not provide any rhythmic or metric quantization (which in flamenco, even when done manually, is a non-trivial task), the relative representation of the onset time is chosen to compensate for strong tempo variation among performances.  

Furthermore, variants of the same melody may be performed in different keys. Singers usually select a key according to their individual pitch range. In order to meaningfully compare the pitch values of a set of performances, we therefore need to perform key normalization. Here, we apply a commonly used method for this task taken from \cite{van2010computational}: One transcription is randomly selected as a reference and histogram of pitch occurrences $h_{Ref}$ is computed. For each of the remaining transcriptions, we compute the pitch histogram in the same way and compute the correlation coefficient with $h_{Ref}$ for different pitch shifts. Finally, for each recording, the pitch shift is applied, which yields the highest correlation coefficient. 

In addition to key transposition, variants of the same melody may exhibit strong rhythmic distortions. We therefore perform a rhythmic normalization as follows. Given $n$ transcriptions of a variant, we select one transcription $A_{ref}$ as a temporal reference to which align all remaining transcriptions. This is realized with the \textit{Needleman-Wunsch} algorithm \cite{needleman1970general} which finds an optimal matching with gaps among two sequences. We use linear interpolation to assign an onset time for notes which have not been matched in the procedure. Figure \ref{fig:align} shows an example of five transcriptions before and after the temporal alignment.

\begin{figure}[h]
\centering
\includegraphics[width=1\columnwidth]{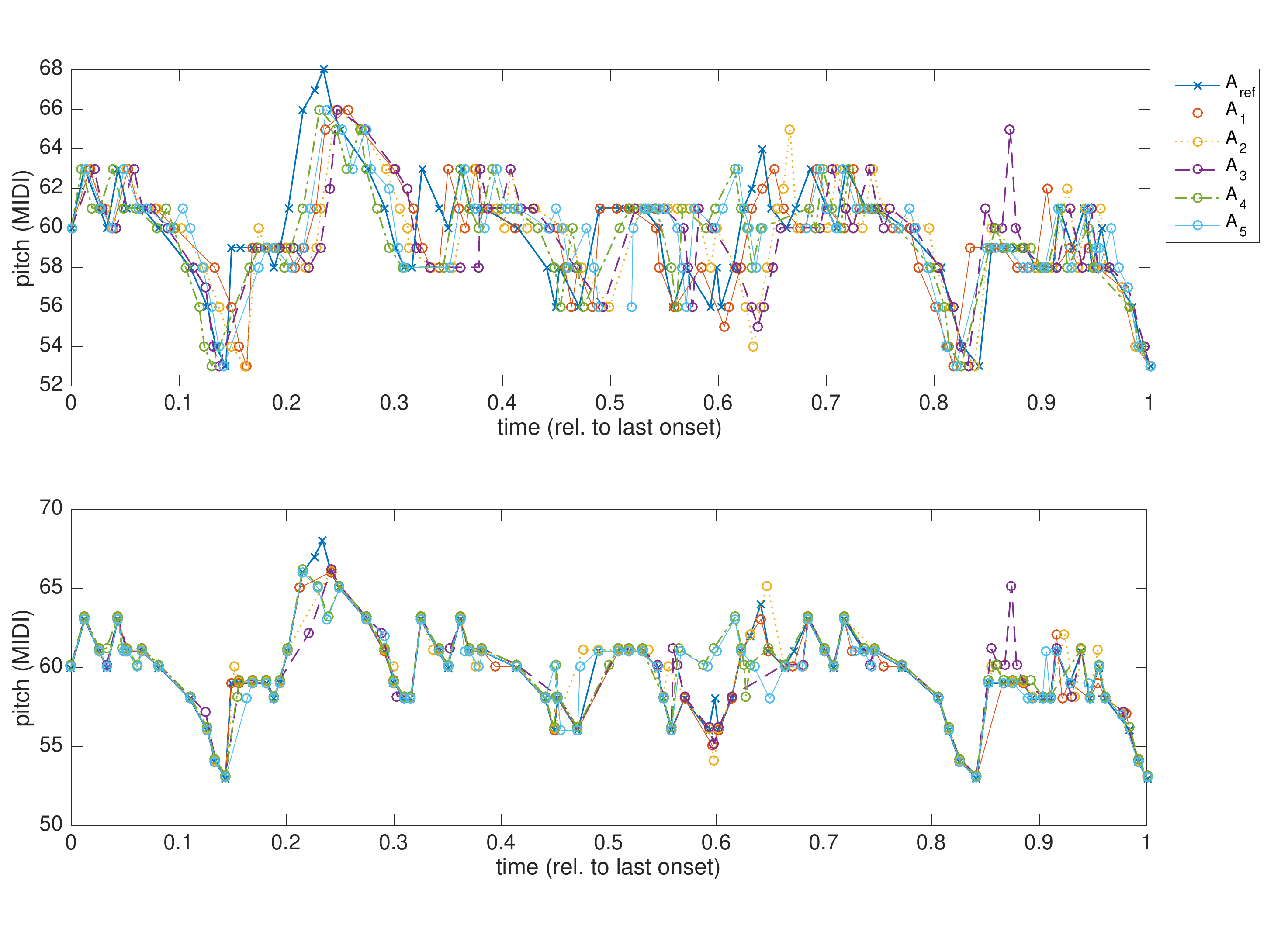}
\caption{Five performance transcriptions before (top) and after (bottom) temporal alignment. 
}
\label{fig:align}
\end{figure}

\section{The geometrical problem} \label{sec:STP}
In order to estimate a melodic template, we aim to find a continuous function, that is, monotone with respect to time, which represents a set of $n$ temporally aligned performance transcriptions. We consider that the data under this study may contain local outliers so some notes could be far from the template. Furthermore, our goal is to quantitatively assess the amount of melodic variation which occurs across performances. 
Geometrically, we view each of $n$ temporally aligned performance transcriptions as a polygonal curve in the plane with the time axis and the pitch axis. Since they are monotone with time, we call them {\em $x$-monotone} curves. 
Typically, the performance transcriptions are discrete and we make an assumption that the curves are polygonal.
The objectives and considerations give rise to a geometrical problem, which we call \emph{Spanning Tube Problem}. 
In this Section we describe an approach for solving the spanning tube problem.

\begin{definition}
Given $a,b\in\R$ and a continuous function $f(x)$ with domain $[a,b]$, we define the $\varepsilon$-tube of $f$, $T(f,\varepsilon),$ as the locus of points $(x,y)$ s.t. $(x,y) \in [a,b]\times[f(x)-\varepsilon, f(x)+\varepsilon]$. An example is illustrated in Figure \ref{fig:tube}.
%$y\in [f(x)-\varepsilon, f(x)+\varepsilon]$, for all $x\in [a,b]$.
\end{definition}

\begin{figure}[t]
\centering
\includegraphics[width=1\columnwidth]{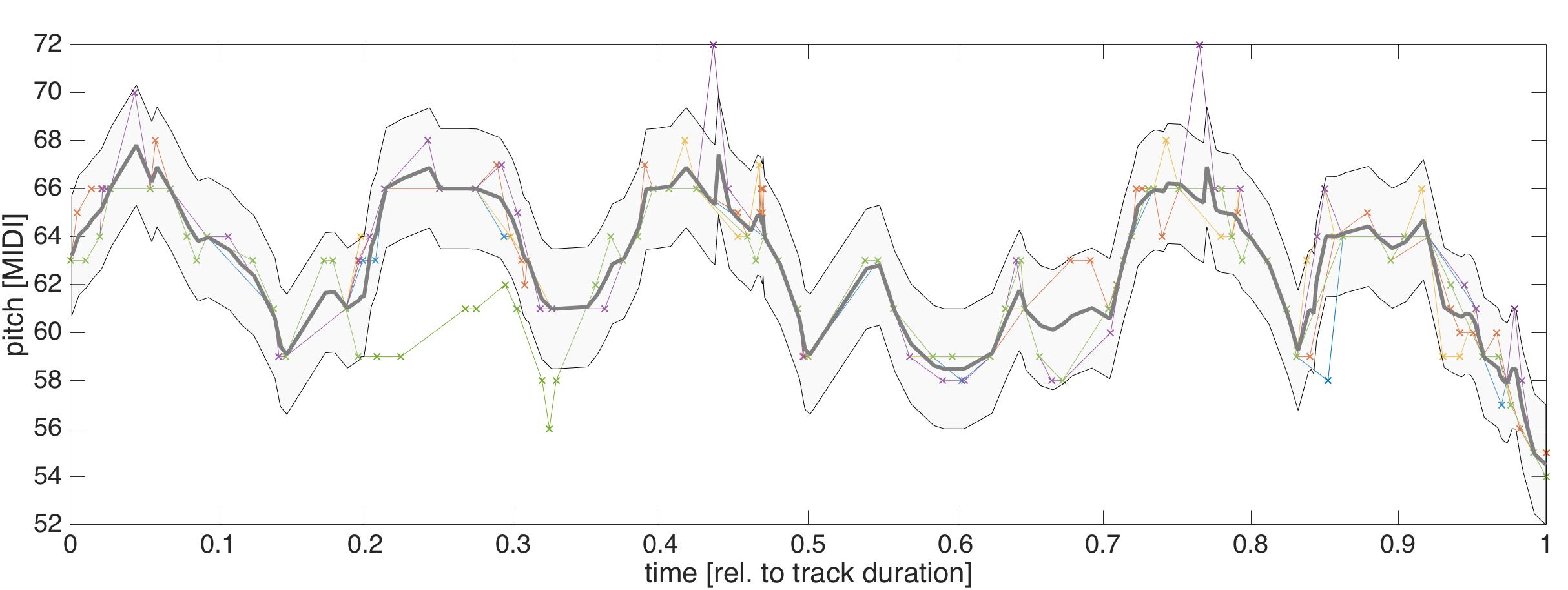}
\caption{The $\varepsilon-$tube for \it{
Fandango de Cala$\tilde{n}$a} (n=5, p=4, $\varepsilon=2.5$).}
\label{fig:tube}
\end{figure}

{\bf The Spanning Tube Problem (STP):}  \emph{Let $a, b \in  \mathbb{R} ,$ with $a<b$;  let $n,m,p \in \mathbb{Z}^+;$ and   for $i=1, \cdots, n,$ let $f_i: \left[a,b\right] \rightarrow \mathbb{R}$ be a piecewise linear function with at most $m$ links. Given $p\leq n,$ find minimum $\varepsilon^*>0$ such that there exists a continuous function $f^*(x)$ fulfilling that, for each $x \in \left[a,b\right]$ the vertical segment of length $2\varepsilon^*$ centered at $(x, f^*(x))$
intersects at least $p$ functions.}

\begin{remark}
Note that the $p$ intersected functions are not necessarily the same in any time and the template $f^*$ has to be continuous and $t$-monotone.
\end{remark}
 
%\vspace{.15cm}

%We call $\varepsilon$-tube to the locus of points $(x,y)$ s.t. $(x,y) \in [a,b]\times[f(x)-\evpsilon, f(x)+\varepsilon]$.

%Given $n$ piecewise linear functions, each with a maximum of $m$ links, and an integer number $p\le n$, find minimum $\varepsilon^*>0$ such that there exists a continuous function $f^*(x)$ such that 
%$\forall t\in[a,b]$ there are at least $p$ functions in $T(f^*,\varepsilon^*)$. 

%\begin{definition}
%Let $f_{\mathrm{up}}(x)$ be the upper and $f_{\mathrm{low}}(x)$ be the lower envelope enclosing the $n$ functions.
%% and $f_{\mathrm{mean}}(x) = \frac{f_{\mathrm{up}}(x)+f_{\mathrm{low}}(x)}{2}$.
%%be the mean of $f_{\mathrm{up}}(x)$ and $f_{\mathrm{low}}(x)$.
%\end{definition}

\begin{definition}
Given a collection of $n$ continuous functions $f_i,$ let $f_{\mathrm{up}}$ be the upper and $f_{\mathrm{low}}$ be the lower envelope enclosing the $n$ functions. Let $f_{median}$ be the point-wise median of the $f_i$'s.
% and $f_{\mathrm{mean}}(x) = \frac{f_{\mathrm{up}}(x)+f_{\mathrm{low}}(x)}{2}$.
%be the mean of $f_{\mathrm{up}}(x)$ and $f_{\mathrm{low}}(x)$.
\end{definition}

\begin{remark} % Inma: I modified the value of \varepsilon^*
For $n=p$, $\varepsilon^*=\max\frac{f_{up}(x)-f_{low}(x)}{2}$ and $f^*(x) = \frac{f_{up}(x)+f_{low}(x)}{2}$ for $x \in[a,b]$ .
\end{remark}

\subsection{Solving the general problem}
The following result can be obtained by using a left-to-right linear sweeping:
\begin{theorem}
The decision problem for STP can be solved in $O(n^2m\log n)$ time.
\end{theorem}

\begin{proof} 
We are given  $n$ functions $f_1(x),f_2(x),\dots,f_n(x),$ $x \in [a, b]$ and two parameters $p\in\{1,2,\dots,n\}$ and $\varepsilon>0$. 
%Let $f_1(x),f_2(x),\dots,f_n(x),$ $x \in [a, b]$ be $n$ given functions and let $\varepsilon$ be a parameter. 
Our task is to decide whether $\varepsilon<\varepsilon^*, \varepsilon=\varepsilon^*$ or $\varepsilon>\varepsilon^*$. This can be accomplished by sweeping the slab between lines with equation $x=a$ and $x=b$ by a vertical line as follows:  we maintain two a sorted lists. Let $L_f$ be the list of functions $f_i(x)\pm\varepsilon$ intersecting the sweep line and let $L_e$ be the list of events where the sweep line stops.
There are two types of events in $L_e$, see Figure~\ref{fig:events}:

{\em Vertex event.} The sweeping line at a vertex event has an equation $x=c$ where one of the functions $f_i(x)$ changes from one linear function to another. Each given function contributes one vertex event to the list of events.

{\em Crossing point event.}
At this event one of the functions $f_i(x)-\varepsilon,f_i(x)+\varepsilon$ intersects one of the functions $f_j(x)-\varepsilon,f_j(x)+\varepsilon$, for some $i<j$. 

\begin{figure}
\centering
\includegraphics[height=6cm]{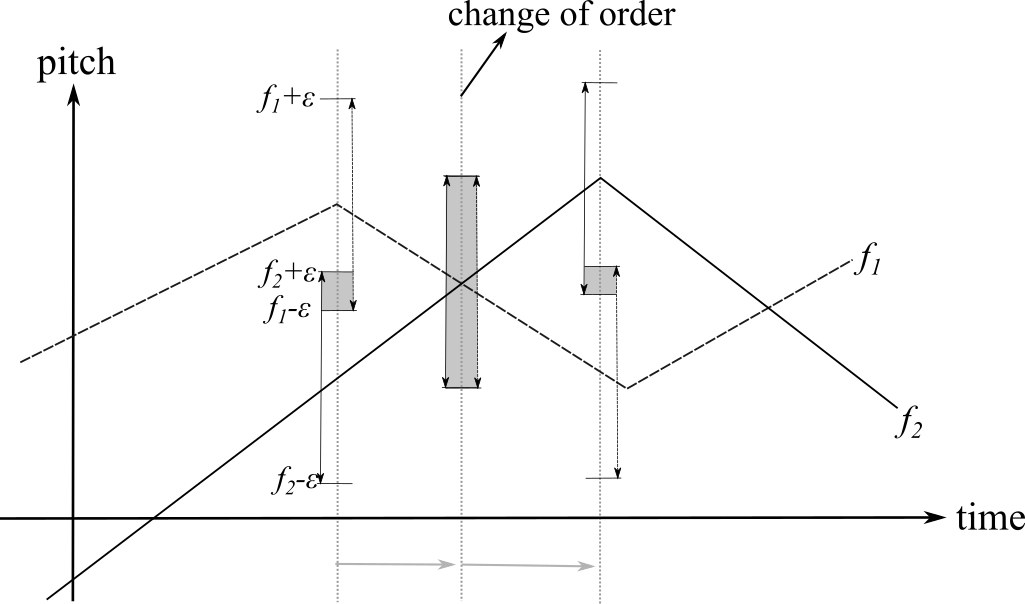}
\caption{The events: Vertex and crossing point events.}
\label{fig:events}
\end{figure}

On the other hand, using $L_f$ we store 
a number $Count(int)$ for each interval $int$ between two consecutive functions crossing the sweep line that indicates the number of  intervals of form 
$[f_i(x)-\varepsilon,f_i(x)+\varepsilon]$ intersecting $int$.
If $Count(int)\ge p$ we also store a boolean $Ind(int)$ indicating whether there is a continuous function $f(x)$ such that $\forall x \in[a,x_0]$ there are at least $p$ functions in $T(f,\varepsilon)$ where $x_0$ is determined by the current position of the sweep line. 

We also maintain the ranges of form 
$[f_i(x)+s_i\varepsilon,f_j(x)+s_j\varepsilon]$, where $s_i,s_j\in \{-1,1\}$, containing values of $y$ such that, for current $x=c$ and any $y$ in the range, there is 
a continuous function $f(x)$ with domain $[a,c]$ and $f(c)=y$.

Note that the continuity of the required function is guaranteed by computing the values of $Count(int)$ and $Ind(int)$.

Observe that $L_e$ can be updated in $O(\log n)$ time and the order of functions $f_i(x)\pm\varepsilon$ is fixed between two events. For any fixed $i$ and $j$, there are at most $O(m)$ crossing points events. 
The total number of events is $O(n^2m)$ and the list of events has size $O(n)$. Since each event can be processed in $O(\log n)$ time, the running time is $O(n^2m\log n)$.

% ==============	QUESTIONS
% \begin{itemize}
% \item \emph{Should not we maintain the intersection of intervals s.t. for any $y$ in the intersection, there is a continuous function connecting $a$ and $y$?}
% \item The number $N$ of functions involved in that intersection must be greater or equal to $p$.
% \item In each step of the sweep we need to count $N$. Is $p$ missing in this counting?
% \item Do not we need that all curves begin and end at the same point?
% \end{itemize}

% \emph{May be we could say some thing like that:
% Given the order of the functions $f_i(x)\pm\varepsilon$, 
% for each $i=1,  2, \cdots, n-p$, compute the intervals
% $INT_i=[f_i(x)-\varepsilon,f_i(x)+\varepsilon]\cap [f_{i+1}(x)-\varepsilon,f_{i+1}(x)+\varepsilon]\cap \cdots
% \cap [f_{n-p}(x)-\varepsilon,f_{n-p}(x)+\varepsilon]$.
% We sweep with a vertical line in which we maintain the $n-p$ consecutive intervals $INT_i$.
% While $INT_i$ is not empty, we can ensure that there exists a function $f$ having at least $p$ function in $T(f,\varepsilon)$.
% (Thus we maintain $O(n-p)$ pointers in the ordered list and $p$ does not appear in the complexity time.) Right?
% }

% ================

\end{proof}

Based on the above result, we can apply bisection for computing an approximate solution. However,
if we dispose of a discrete set of candidate  values for $\varepsilon^*$, an exact solution can be found. The following lemma gives us the discrete set of candidate values for the optimization problem (STP).
%which are shown in Figure \ref{fig:scheme}.

\begin{definition}
We define a vertex of a piecewise linear function as a vertex-point and an intersection between two functions as an intersection-point.
\end{definition}

\begin{figure}[htb]
\centering
\includegraphics[width=1\columnwidth]{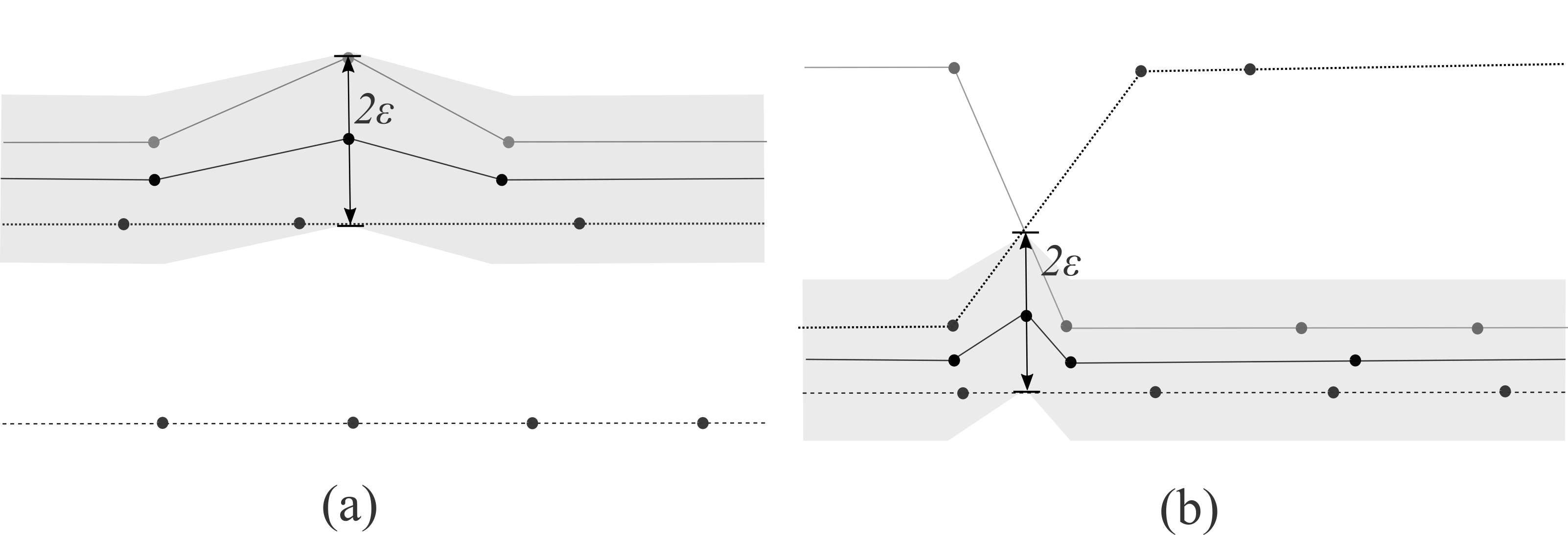}
\caption{Two scenarios for determining the candidate values: (a) $2\varepsilon$ corresponds to the distance between a vertex of a function $f_i$ and the value of another function $f_j$ at the same time instance; (b) $2\varepsilon$ corresponds to the distance between the intersection of two functions $f_i$ and $f_k$ and the value of other $f_j$ at the same time.}\label{fig:scheme}
\end{figure}

\begin{lemma}\label{lma1}
For $n > 2$ and $1<p<n$, let $\varepsilon^*$ be the optimal value for the optimization problem. Then, $2\varepsilon^*$ is equal to one of the values given by
%\begin{itemize}
%\item{
the vertical distance between a vertex- (Figure \ref{fig:scheme} (a)) or an intersection-point (Figure \ref{fig:scheme} (b)) and its $(p-1)$ or $p$-nearest function.
%} 
% \item{the vertical distance between an intersection-point and $(p-1)$-nearest function .}
%\end{itemize}
\end{lemma}

\begin{proof} 
Let $\varepsilon^*$ be the minimum value such that there exists a continuous function $f^*(x)$ such that $\forall x\in[a,b]$ there are at least $p$ functions in $T(f^*,\varepsilon^*)$.
Then, there are two functions $f_i$ and $f_j$ in $T(f^*,\varepsilon^*)$ and $x_0 \in [a,b]$ such that $f_i(x_0)=f^*(x_0) + \varepsilon^*$ and $f_j(x_0)=f^*(x_0) - \varepsilon^*$.
Otherwise, we can decrease the width  of the tube contradicting optimality. It is easy to prove that the intersection points between the functions $f_i$ and $f_j$ and the boundary of the tube must have the same abscissa. 
There are three cases: the point $u=(x_0, f_i(x_0))$ is a vertex-point, a crossing point, or none of the above. Analogously for $v=(x_0, f_j(x_0)).$

\begin{figure}[h]
\centering
\includegraphics[scale=0.9]{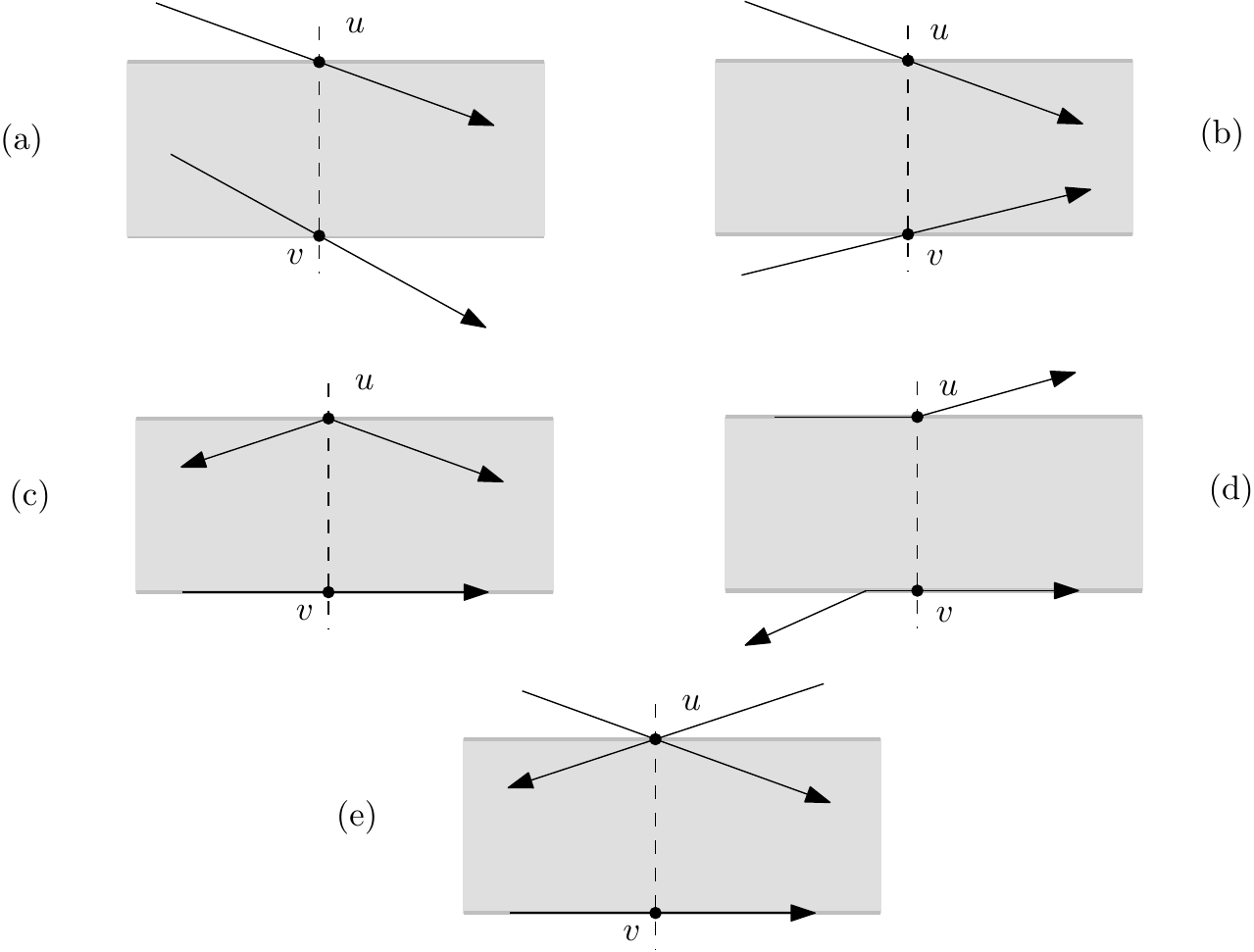}
\caption{(a) and (b): $u$ and $v$ are neither vertex nor crossing points. (c) and (d): $u$ is a vertex point. (e): $u$ is a crossing point.}
\label{fig:lemma8}
\end{figure}

Suppose that $u$ and $v$ are neither a vertex nor a crossing point. 
Then up to symmetry, we have two cases as in Figure~\ref{fig:lemma8} (a) and (b). 
In both cases, the width of the tube can be slightly decreased maintaining continuity. This fact contradicts the optimality condition.  
In case shown in Figure~\ref{fig:lemma8} (a) we can move the endpoints $u$ and $v$ on $f_i$ and $f_j$ respectively maintaining vertical the segment $uv$ in the direction where the length of the vertical segment decreases. Then, there exists an instant in which a vertex or crossing point is found, so this case reduces to the vertex or crossing point event. 
Note that even when $f_i$ and $f_j$ are parallel, theses events are found if we maintain at least $p$ functions inside the tube.
A similar reasoning can be done in the second case (in Figure~\ref{fig:lemma8} (b)): since the part of the slab to the left from the vertical segment $uv$ must have at least $p$ functions, then in the right part there are at least $p+2$ functions so there exists a narrower tube containing at least $p$ functions.

Consider now the case in which $u$ is a vertex-point of $f_i(x)$ (w.l.o.g. $u$ lies on the top boundary of the tube).  
Two cases arise, depending on the slopes of the neighboring segments, see (c) and (d) in Figure~\ref{fig:lemma8}. 

{\em Case 1}. 
The function  $f_i(x)$ is in the tube for $x$ in some neighborhood of $x_0$ as in in Figure~\ref{fig:lemma8} (c).
Then $f_j(x)$ is the $p-1$-nearest function of $u$.

{\em Case 2}. 
For $\delta >0$ and $x \in (x_0, x_0 + \delta)$ or $x \in ( x_0 - \delta, x_0)$,  $f_i(x)$ is not in the tube as in Figure~\ref{fig:lemma8}(d).  
The function $f_i(x)$ on the top boundary comes out the tube at time $x_0$ and the function $f_j(x)$ is the $p$-nearest function.

Finally, the case in which $u$ is a intersection-point can be easily reduced to the vertex-point type, see Figure~\ref{fig:lemma8}(e).
\end{proof}

% \begin{figure}[h]
% \centering
% \includegraphics[width=0.8\columnwidth]{schemeNew.eps}
% \caption{Cases of Lemma 8. The shaded region is the tube. }\label{fig:lemma8}
% \end{figure}

\begin{theorem} \label{general}
For $n > 2$ and $p>1$,  the optimization problem can be solved in $O(n^2m\log n\log nm)$ time.
\end{theorem}
\begin{proof}
The steps of the algorithm are:
(I) Compute the candidate values;
(II) Sort the candidates and 
(III) Perform binary search using the decision algorithm.
Step (I) can be done by sweeping the arrangement of the $n$ functions with a vertical line. In this status line we maintain the order of the intersection of the line with the functions (it can be vertex-points, intersection-points and functions) so that we can compute in $O(1)$ the corresponding $p-1$- or $p$-nearest functions (one above and other below) for every candidate. The overall complexity of the sweep is $O(n^2m\log n)$.
Step (I) can be done in $O(n^2m\log n)$ time, Step (II) requires $O(n^2m\log nm)$ time and Step (III) in $O(n^2m\log n\log nm)$ time.
\end{proof}

\subsection{An efficient solution for a particular case}

We have the following result for a particular case:

\begin{lemma}\label{32}
Set $n = 3, p = 2$ and let $\varepsilon^*$ be the optimal value for the optimization problem. Then the value of $2\varepsilon^*$ is one of these values:
\begin{itemize}
\item a local  maximum value of $f_{up}-f_{median}$ or
\item a local maximum value of $f_{median}-f_{low}$ or 
\item a local  minimum value of $f_{up}-f_{low}$.
\end{itemize}
\end{lemma}  
\begin{proof}
For $\varepsilon>0$, let $R_1(\varepsilon)=T(f_{up},\varepsilon ) \cap T(f_{mean},\varepsilon)$ and let $R_2(\varepsilon)=T(f_{low},\varepsilon ) \cap T(f_{mean},\varepsilon)$. Observe that, a solution exits for $\varepsilon$ if and only if there exists a monotone function inside the region $R_1(\varepsilon)\cup R_2(\varepsilon)$.
Now, a  solution exists for $R_1(\varepsilon^*) \cup R_2(\varepsilon^*)$ but no within $R_1(\varepsilon)\cup R_2(\varepsilon)$ for any  $\varepsilon < \varepsilon^*$. 
Note that $R_i(\varepsilon)\subset R_i(\varepsilon^*)$ for $i = 1, 2$. We have 3 cases:
\begin{itemize}
\item Two points in $R_1(\varepsilon^*)$ are connected by a function  but are not connected in $R_1(\varepsilon)$. Then, for some $x^*\in (a,b)$, we have $f_{up}(x^*)-\varepsilon   = f_{mean}(x^*) + \varepsilon$ and $f_{up}(x)-\varepsilon < f_{mean}(x)+ \varepsilon$ for all $x\not= x^*$ in some neighborhood of $x^*$.
Therefore $f_{up}-  f_{mean}$ has a local maximum value at $x^*$ with value $2\varepsilon^*$.
\item Two points in $R_2(\varepsilon^*)$ are connected by a function  but are not connected in $R_2(\varepsilon)$. Then, for some $x^*\in (a,b)$, we have $f_{mean}(x^*)-\varepsilon   = f_{lower}(x^*) + \varepsilon$ and $f_{mean}(x)-\varepsilon < f_{low}(x)+ \varepsilon$ for all $x\not= x^*$ in some neighborhood of $x^*$.
Therefore $f_{mean}-  f_{low}$ has a local maximum value at  $x^*$ with value $2\varepsilon^*$.
\item Two points, one in $R_1(\varepsilon^*)$ and one in $R_2(\varepsilon^*)$, are connected in  $R_1(\varepsilon^*) \cup R_2(\varepsilon^*)$ but are not within  $R_1(\varepsilon)\cup R_2(\varepsilon)$.  Then, for some $x^*\in (a,b)$, we have $f_{up}(x^*)-\varepsilon   = f_{low}(x^*) + \varepsilon$ and $f_{up}(x)-\varepsilon > f_{low}(x)+ \varepsilon$ for all $x\not= x^*$ in some neighborhood of $x^*$. Therefore $f_{up}-  f_{low}$ has a local minimum value at  $x^*$ with value $2\varepsilon^*$.
\end{itemize}
\end{proof}

By Lemma \ref{32}, we have $O(m)$ candidates and performing binary search the optimal solution can be computed in $O(m\log m)$ time by using the decision algorithm. We will improve the running time by proving some properties first.
%In the following we prove some technical properties needed in the proof above:

\begin{definition}
Let $x_0=a,x_1,\dots,x_t=b$ be the events defined by vertices and intersection points of the functions. Consider a slab between two events at $x_i$ and at $x_{i+1}$.
Suppose that a spanning tube $T$ in $[x_0, x_{i+1}]$ covers $f_{up}(x_i),f_{median}(x_i)$ and $f_{low}(x_{i+1}),f_{median}(x_{i+1})$.
We say that the tube make a {\em transition HL} in the $i$th slab. 
\end{definition}

\begin{lemma}[Two trapezoids]
Suppose a spanning tube $T$ make a transition HL in an $i$th slab.
Then there exists $x\in [x_i,x_{i+1}]$ such that $f_{up}(x),f_{median}(x)$ and $f_{low}(x)$ are covered by $T$.
Furthermore, the width of $T$ is at least 
\begin{equation} \label{2trap}
w=\max(\varepsilon_1^i,f_{up}(x)-f_{low}(x),w_2^{i+1})
\end{equation}
and there exist a tube of width $w$ using two trapezoids in the slab, see Figure \ref{fig:2trap}.
\end{lemma}

\begin{figure}[h]
\centering
\includegraphics[width=8cm]{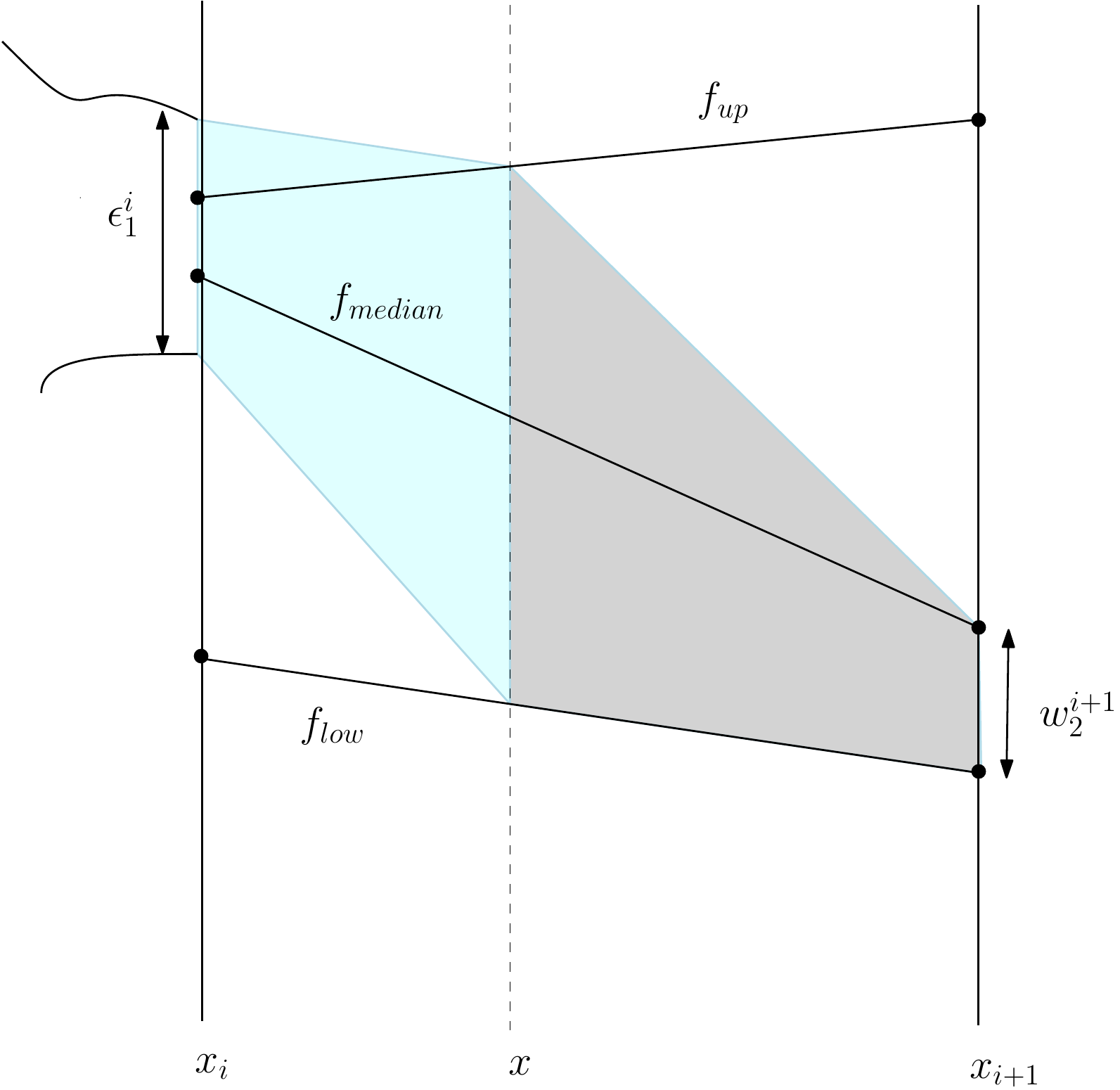}
\caption{Two trapezoids: The tube make a transition in the ith slab.}\label{fig:2trap}
\end{figure}

\begin{proof}
The existence of $x$ can be seen using continuity of the tube.
Since all 3 functions are linear in the slab, the first trapezoid covers $f_{up}$ and $f_{median}$.
Similarly the second trapezoid covers $f_{median}$ and $f_{low}$. 
The lemma follows since the width of the tube formed by two trapezoids is $w$.
\end{proof}

\begin{lemma}[Transition] \label{l:trans}
The smallest width of a spanning tube $T$ making a transition HL in an $i$th slab is equal to
$\max(\varepsilon_2^i,\min(w^i,w^{i+1}),w_1^{i+1})$.
The vertical for the transition can be decided by comparing $w^i$ and $w^{i+1}$.
Thus, there exists an optimal tube having \emph{all} transitions at the events only.
\end{lemma}

\begin{proof}
This can be seen by changing $x$ between $x_i$ and $x_{i+1}$.
The smallest value of the second term in Equation \ref{2trap} is $\min(w^i,w^{i+1})$.
Thus, if $f_{up}(x)$ and $f_{low}(x)$ are not parallel, we can decide the line for transition ($x=x_i$ or $x=x_{i+1}$ using the smallest value of $w^i$ and $w^{i+1}$). If they are parallel, one can choose line $x=x_i$ (or $x=x_{i+1}$) for the transition.
\end{proof}

Now, we are ready to show that we can solve the problem in linear time.

\begin{theorem} \label{n3p2}
The optimization problem for $n = 3$ and $p = 2$ can be solved in  $O(m)$ time.
\end{theorem}

\begin{proof}

%(vertices or intersections points between two functions). Figure \ref{fig:linear} shows an instance of the updating in a slap between two events. Since we have $O(m)$ events, the overall complexity is linear. Other cases of the basic problem can also be solved in $O(1)$ time.

\begin{figure}
\centering
\includegraphics[height=7cm]{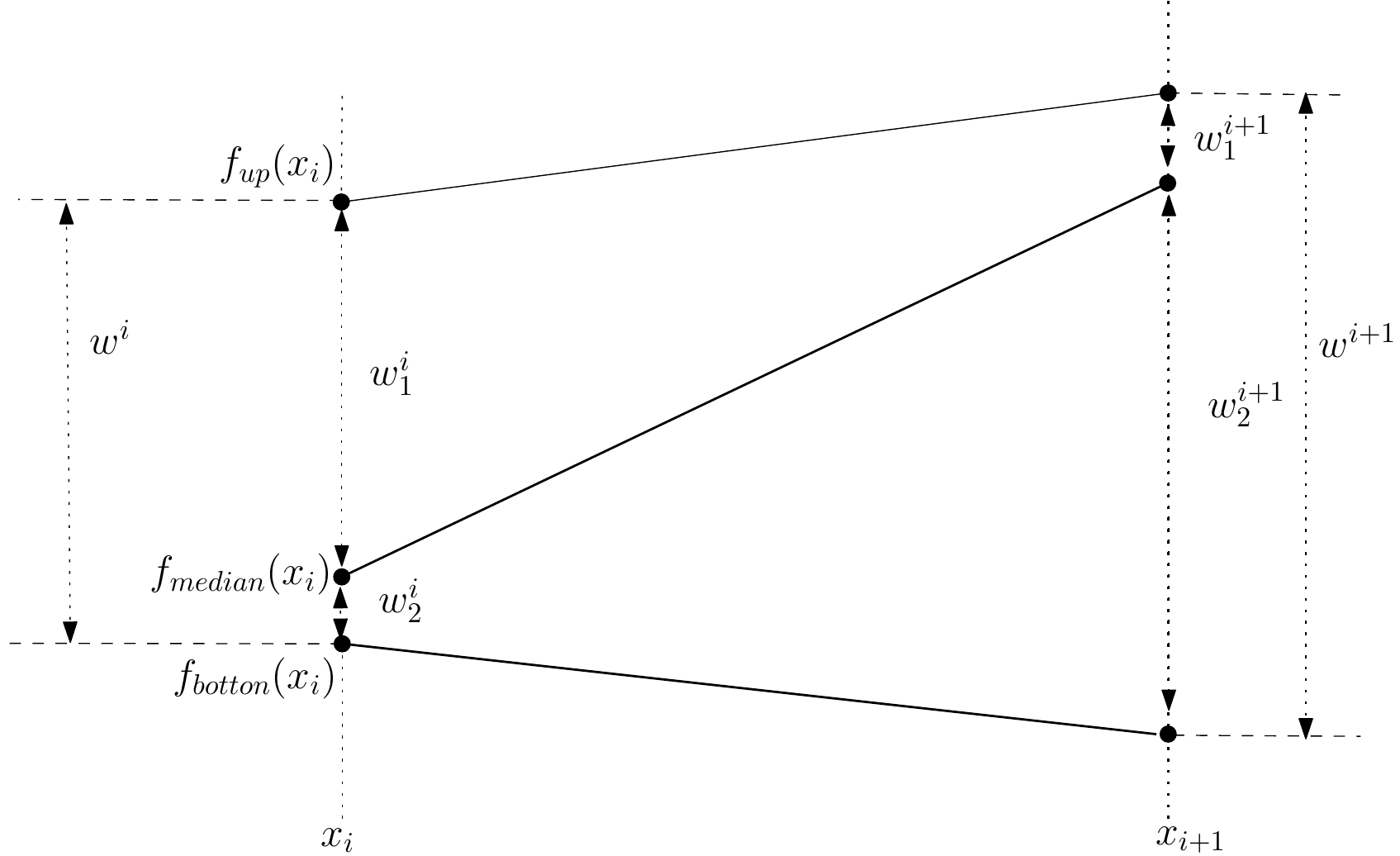}
\caption{An instance for updating the optimal tube in the slab between the events $x_i$ and $x_{i+1}$ for $n=3, \, p=2$.}
\label{fig:linear}
\end{figure}

Let $x_0=a,x_1,\dots,x_t=b$ be the events defined by vertices and intersection points of the functions.
For each event at $x=x_i$, we compute $\varepsilon^i_1$ and $\varepsilon^i_2$ for the spanning tubes in $[x_0,x_i]$ covering $[f_{up}(x_i),f_{median}(x_i)]$
and $[f_{median}(x_i),f_{low}(x_i)]$, respectively.
Let $w_1^i=f_{up}(x_i)-f_{median}(x_i)$ and $w_2^i=f_{median}(x_i)-f_{low}(x_i)$.
At the beginning, $\varepsilon_j^0=w_j^0    $ for $j=1,2$. 
Then $\varepsilon_1^{i+1}$ can be computed using two cases where the tube covers (i) $f_{up}(x_i)$ or (ii) $f_{low}(x_i)$. Then $\varepsilon_1^{i+1}$ is the minimum of the two. The second value is computed using Lemma \ref{l:trans} and $w^i=f_{up}(x_i)-f_{low}(x_i)$, see Figure \ref{fig:linear}.
Thus,
\begin{equation} 
\label{eq1}
\varepsilon_1^{i+1} = \min (\max (\varepsilon_1^{i},w_1^{i+1}),\max(\varepsilon_2^i,\min(w^i,w^{i+1}),w_1^{i+1})).
\end{equation}
Similarly,
\begin{equation} 
\label{eq2}
\varepsilon_2^{i+1} = \min (\max (\varepsilon_2^{i},w_2^{i+1}),\max(\varepsilon_1^i,\min(w^i,w^{i+1}),w_2^{i+1})).
\end{equation}
Finally, $\varepsilon^*=\min(\varepsilon_1^t,\varepsilon_2^t)$.

\medskip

{\em Running time}.
For each event $i+1$, the values of 
$w^{i+1}, w_1^{i+1}, w_1^{i+2}$ (see Figure \ref{fig:linear}) and the values of 
$\varepsilon_1^{i+1}, \varepsilon_2^{i+1}$ (by Equations (\ref{eq1}) and (\ref{eq2})) can be computed in constant time. Since we have $O(m)$ events, the overall running time is linear. 
\end{proof}

\begin{remark} Note that Theorem \ref{n3p2} improves the complexity obtained in Theorem \ref{general} for $n=3$. 
\end{remark}

\section{Case study: Quantifying melodic variation} \label{sect:case}
In this Section, we exemplify the use of the proposed method in a comparative performance analysis of the four \textit{fandango} styles investigated in this study. In particular, we use the proposed system to quantify the amount of melodic variation the skeleton is subjected to during performance, compare across styles and analyze the local variation on a phrase-level. 

From a music theoretic standpoint we know that among the four styles under study, the \textit{Fandangos de Cala{\~n}a} and the \textit{Fandangos de Valverde} are close to their folkloric origin and performers tend to largely preserve the melodic skeleton during performance. In the \textit{Fandangos Valientes de Huelva} and the \textit{Fandangos Valientes de Alosno} on the other hand, performers tend to use heavy melodic ornamentation as an artistic asset, resulting in a more distorted skeleton. 

The proposed method allows us to quantify the amount of melodic variation occurring in a set of performance transcriptions by fixing the parameter $\varepsilon$ and determining in a decision problem the maximum percentage of performances $\frac{p}{N}$ which can be enclosed by the tube. Here, we compute this value for each of the four styles under study, varying $\varepsilon$ between $1.0$ and $3.0$. 

\begin{table}
\caption{Fraction of performances enclosed by the tube for different styles.}
\label{tab:results2}
\centering
\begin{tabular}{l|c | c| c|}
Style & $\varepsilon=1.0$ & $\varepsilon = 2.0$ & $\varepsilon = 3.0$\\
\hline
Fandango de Cala{\~n}a &  0.2 &  0.5 & 0.8\\
Fandango de Valverde  &   0.1 & 0.3 & 0.5\\
Fandango Valiente de Alosno &  0.1 & 0.3 & 0.3\\
Fandango Valiente de Huelva &  0.0 & 0.1 & 0.1\\
\end{tabular}
\end{table}

The largest differences among styles can be observed for $\varepsilon = 3.0$. Note, that in this case, the tube covers $6$ semitones, which corresponds to half an octave. Consequently, melodic segments outside the tube correspond to relatively large deviations from the basic melodic contour. For this case, the results (Table \ref{tab:results2}) confirm the musicological considerations described above: For the \textit{Fandangos de Cala{\~n}a}, 80\% of the analyzed performance transcriptions can be enclosed by the tube with $\varepsilon = 3.0$, indicating that performers largely follow the underlying skeleton. The value for the \textit{Fandangos de Valverde} is slightly lower with $\frac{p}{N}=0.5$. The two \textit{valiente} styles show a significantly higher amount of melodic variation. For the \textit{Fandangos Valientes de Huelva} only 10\%, and for the \textit{Fandangos Valientes de Alosno} only 30\% of the transcriptions are enclosed in the tube. Figure \ref{fig:ex} shows the aligned transcriptions together with the computed templates using the values for $\frac{p}{N}$ from table \ref{tab:results2}. A similar trend can be observed for smaller values of $\varepsilon$.

\begin{figure}[th]
\centering
\includegraphics[width=.8\columnwidth]{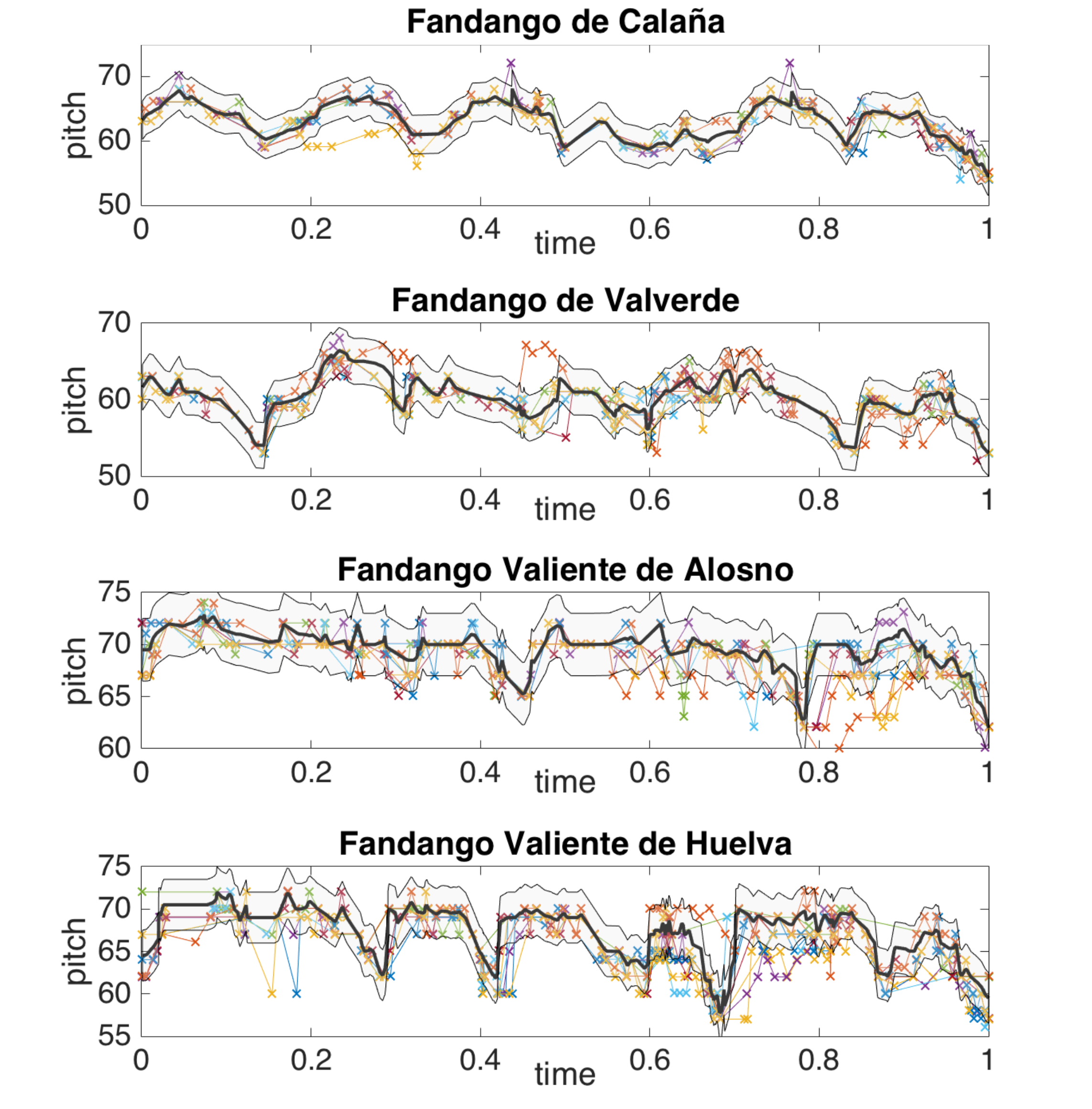}
\caption{Aligned transcriptions and computed skeleton for the four \textit{fandango} styles under study for $\varepsilon=3.0$.}\label{fig:ex}
\end{figure}

However, analysing the transcriptions in relation to the maximum spanning tube (Figure \ref{fig:ex}), reveals that there exist local differences in the amount of occurring variation. For example, for the \textit{Fandangos Valientes de Alosno}, all transcriptions are located inside the $3 \varepsilon$-tube from the beginning of the melody until approximately $0.3$ on the relative time axis. In order to obtain a finer granularity of the amount of occurring variation, we therefore repeat the previous experiment on a phrase level. More precisely, each of the recordings contains $6$ musical phrases, which we manually annotated. Fixing $\varepsilon=3.0$, we solve the optimization problem and compute $\frac{p}{N}$ for each phrase separately. 

The results in Table \ref{tab:results3} show, that with exception of the \textit{Fandango de Valverde}, the last phrase tends to exhibit a high amount of variation compared to other phrases. This phrase, which is referred to as \textit{ca{\'i}da} in flamenco jargon, represents the highlight of the flamenco performance, and consequently, it is likely that performers use a larger amount of ornamentation and variation as an expressive asset. For all styles, we furthermore observe a relatively high amount of variation for the fourth phrase. 

These observations are of interesting from a musicological viewpoint and give rise to several lines of study related to melodic variation and expressiveness in flamenco music.

\begin{table}
\caption{Fraction of performances enclosed by the tube per phrase for different styles and $\varepsilon=3.0$.}
\label{tab:results3}
\centering
\begin{tabular}{l|c | c | c | c | c | c |}
Style & phrase 1 & phrase 2 & phrase 3 & phrase 4 & phrase 5 & phrase 6\\
\hline
Fandango de Cala{\~n}a &  0.5 &  0.7 & 0.7 & 0.5 &0.7  & 0.4\\
Fandango de Valverde  &   0.5 & 0.5 & 0.2 & 0.2& 0.3 &0.5\\
Fandango Valiente de Alosno &  0.3 & 0.4 & 0.2 &0.1 &  0.5&0.1\\
Fandango Valiente de Huelva &  0.3 & 0.7 & 0.5 & 0.2 & 0.3&0.1\\
\end{tabular}
\end{table}

\section{Conclusions and Future Research} \label{sect:conclusion}

Motivated by musical properties typically encountered in the analysis of oral traditions, we introduced in this paper a new geometric optimization problem, the Spanning Tube Problem (STP). 
We model $n$ melodies as polygonal curves in the time-pitch space with at most $m$ vertices. The aim is to compute a new polygonal curve, the template, which fits a fixed number $p$ of similar items. The particular challenge here is, that the parameter $p$ corresponds to an amount of melodies, but does not refer to a specific set of melodies. We solve the optimization problem by performing binary search on a discrete set of candidate values, for which we solve the corresponding decision problem.
The obtained time complexity is $O(n^2m\log n\log nm)$. We also prove that the particular case $n=3$, $p=2$ can be solved in linear time with a more elaborate approach.
Finally, we perform an experimental study with flamenco melodies demonstrating how the resulting STP can be employed in a comparative performance analysis to quantify the amount of variation among a set of melodies.

There are several immediate suggestions for further research. 
\begin{itemize}
\item The first issue is related to the complexity. Can the asymptotic complexity of the problem be improved? For large data sets, the time complexity of our approach is roughly cubic and we ask if a more detailed study allows us to improve the algorithm, as we did for the case $n=3$ and $p=n-1=2$. Since $p$ could be close to $n$, it is even interesting to find a $O^*(n(n-p)m)$-time algorithm.

\item In the STP, the width of the tube is the same at any point in time. However, we observed in the transcriptions that the melodic variability is not constant. Consequently, a future research could target a STP with variable width. Observe that the simple idea of intersecting the optimal tube with the polygon between the upper and lower hulls of the functions gives us a more accurate visualization of the variability.

\item Another possible variant is to restrict the number of vertices of the template. In fact, for highly ornamented melodies, the template appears to be very complex and a simpler prototype could better model the underlying melodic movement.

\item Another interesting question is to efficiently solve the reverse problem: Given an $\varepsilon>0$, compute the maximum number $p^*$ of melodies that can be captured by an $\varepsilon$-tube. Note, that our algorithm solves this problem in $O(n^2m\log n \log p)$ using binary search. Is it possible to improve the running time using a different approach?

\item Finally, the extension of our problem to 3D leads to an interesting task: the recognition of user-defined temporal gestures from tactile interfaces.
Given a set of $n$ curves (gestures) $f_i$ in the $XY$ plane, we want to compute a new continuous curve approximating the set of curves $f_i$. %Figure~\ref{fig:signature} shows the two-dimensional curve of a signature.
Note that the gestures $f_i$ can be represented by the orthogonal projection of three dimensional curves in the $XYT$ space that are monotone on time. Consequently, this problem can be seen as a generalization of the Spanning Tube Problem, see Figure~\ref{fig:signature3D}.
The definition in 2D can be extended to 3D as follows:

{\bf The 3D STP}:  \emph{Let $a, b, c ,d, t \in  \mathbb{R} ,$ with $a<b, \, c<d$ and $t>0$;  let $n,m,p \in \mathbb{Z}^+;$ and   for $i=1, \cdots, n,$ let $f_i: \left[0, t\right] \rightarrow \mathbb{R}^2$ be a T-monotone piecewise linear function with at most $m$ links. Given $p\leq n,$ find minimum $\varepsilon^*>0$ such that there exists a T-monotone continuous function $f^*(t)=(f_1(t), f_2(t))$ fulfilling that, for each $t \in \left[0, T\right]$  the ball (for a $L_p$ distance ($L_2, L_1, L_{\infty}, \cdots$)) of radius $\varepsilon^*$ centered at $(t, f^*(t))$ intersects at least $p$ functions.}

% \begin{figure}[ht]
% \centering
% \includegraphics[width=0.4\columnwidth]{signature.jpg}
% \caption{Is it my signature?}\label{fig:signature}
% \end{figure}
\begin{figure}[ht]
\centering
\includegraphics[width=0.5\columnwidth]{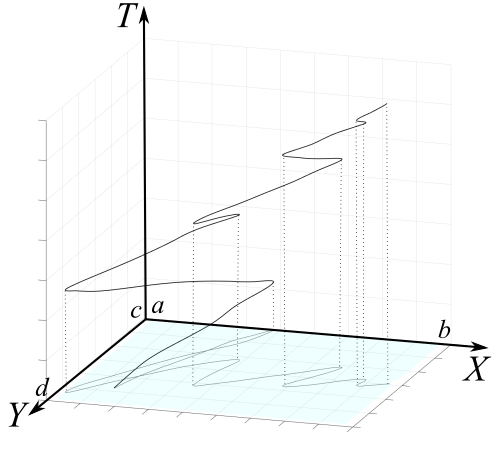}
\caption{The orthogonal projection of $f$ onto the XY plane is the gesture.  $f$ is a T-monotone curve.}\label{fig:signature3D}
\end{figure}

\end{itemize}

 \bibliographystyle{abbrv}
\bibliography{papers}

\end{document}